\documentclass[12pt,final]{article}

\usepackage{color}
\usepackage{amsmath,amsthm,amsfonts,amsthm,amssymb}
\usepackage{mathrsfs} %needed for script fonts
\usepackage{fullpage}
\usepackage{graphicx}
\usepackage{float}

\newtheorem{lemma}{Lemma}
\newtheorem{remark}{Remark}
\newtheorem{proposition}{Proposition}
\newtheorem{theorem}{Theorem}
\newtheorem{definition}{Definition}

\renewcommand\nu n

\newcommand{\N}{\mathbb N}

\definecolor{gray}{rgb}{0.9,0.9,0.9}

\usepackage{amssymb,amsmath,mathrsfs,geometry,
pdfsync}
\usepackage{graphicx}
\newcommand{\ignore}[1]{{}}

\usepackage{xcolor}
\usepackage[normalem]{ulem}
\usepackage{cancel}

\definecolor{brown}{rgb}{0.6,0.2,0.2}

\usepackage{hyperref}
\usepackage{esint}

\newlength{\overwritelength}
\newlength{\minimumoverwritelength}
\newcommand{\overwrite}[3][red]{%
  \settowidth{\overwritelength}{$#2$}%
  \ifdim\overwritelength<\minimumoverwritelength%
    \setlength{\overwritelength}{\minimumoverwritelength}\fi%
  \stackrel
    {%
      \begin{minipage}{\overwritelength}%
        \color{#1}\centering\small #3\\%
        \rule{1pt}{9pt}%
      \end{minipage}}
    {\colorbox{#1!50}{\color{black}$\displaystyle#2$}}}

\usepackage{color}

\renewcommand{\ignore}[1]{}

\makeatletter
\renewcommand*{\@fnsymbol}[1]{\ifcase#1\or*\else\@arabic{\numexpr#1-1\relax}\fi}
\makeatother

\begin{document}

\title{DISSIPATIVE SCALE EFFECTS IN STRAIN-GRADIENT PLASTICITY: \\ THE CASE OF SIMPLE SHEAR\thanks{This work was partially supported by Sapienza Award 2013 Project ``Multiscale PDEs in Applied Mathematics'' (C26H13ZKSJ) and by INdAM-GNFM through the initiative ``Progetto Giovani''.}}

\author{Maria Chiricotto\thanks{Institute for Applied Mathematics, University of Heidelberg, Im Neuenheimer Feld 294, 69120 Heidelberg, Germany -
chiricotto@math.uni-heidelberg.de} \and Lorenzo Giacomelli\thanks{SBAI Department, Sapienza University of Rome, Via Scarpa 16, 00161 Roma, Italy - lorenzo.giacomelli@sbai.uniroma1.it} \and Giuseppe Tomassetti\thanks{Department of Civil Engineering and Computer Science Engineering, University of Rome ``Tor Vergata'', Via Politecnico 1, 00133 Roma - tomassetti@ing.uniroma2.it}}

\maketitle

\begin{abstract}
We analyze dissipative scale effects within a one-dimensional theory, developed in [L. Anand et al. (2005) J. Mech. Phys. Solids 53], which describes plastic flow in a thin strip undergoing simple shear. We give a variational characterization of the {\emph{ yield (shear) stress}} --- the threshold for the inset of plastic flow --- and we use this characterization, together with results from [M. Amar et al. (2011) J. Math. Anal. Appl. 397], to obtain an explicit relation between the yield stress and the height of the strip. The relation we obtain confirms that thinner specimens are stronger, in the sense that they display higher yield stress.
\end{abstract}



%\pagestyle{myheadings}
%\thispagestyle{plain}
%\markboth
%{M. CHIRICOTTO, L. GIACOMELLI, G. TOMASSETTI}
%{DISSIPATIVE SCALE EFFECTS IN GRADIENT PLASTICITY}

\section{Introduction}

A number of experiments have shown that conventional plasticity fails to capture the size-dependent behavior of metallic specimens undergoing plastic flow in the size range below 100 microns, with smaller samples being, in general, stronger (see \cite{Hutch2000IJSS} for a review).

\smallskip

Substantial theoretical work has been carried out to extend conventional plasticity to the micron scale. It is acknowledged that size effects observed in metallic samples are associated to the inhomogeneity of plastic flow \cite{ashby1970deformation}, a fact that motivates a number of \emph{strain-gradient plasticity} theories, starting with Ref.~\cite{DilloK1970IJSS}.

\smallskip

In the so-called {\emph{non-local} or } \emph{high-order theories}, the flow rule that governs the evolution of plastic strain is a partial differential equation which requires the specification of appropriate boundary conditions. The first of such theories was proposed by Aifantis \cite{Aifan1984JoEMaT}; the vast majority of subsequent high-order theories were derived using the virtual-power principle, by taking into account power expenditure by higher-order stresses that are work-conjugate to the plastic-strain gradient \cite{Barde2006JMPS,
FleckH2001JMPS,Gudmu2004JMPS,Gurti2004JMPS,GurtiA2005JMPS}.

\smallskip

Apparently, the theories developed by Gurtin and Anand \cite{Gurti2004JMPS,GurtiA2005JMPS} are those that have inspired most mathematical work. One of its distinctive aspects is that the full plastic distortion (the sum of a symmetric plastic strain and a skew--symmetric plastic spin) is accounted for. In  particular, the issues of existence and uniqueness of solutions for strain-gradient plasticity with plastic spin, as considered in Ref.~\cite{Gurti2004JMPS}, has been addressed in Ref.~\cite{BertsDGT2011DCDS} in the case of two-dimensional setting of anti-plane shear, and in Refs.~\cite{ebobisse2010existence}, \cite{neff2008uniqueness}, and \cite{NeffCA2009MMMAS} in the full three-dimensional setting. The model for plastically-irrotational materials proposed in Ref.~\cite{GurtiA2005JMPS} was studied in Ref.~\cite{ReddyEM2008IJP}. Of particular importance for the present paper are the existence theorems for strain-gradient plasticity based on the notion of energetic solution, which have been proved both in the small--strain \cite{GiacoL2008SJMA} and in the large--strain \cite{MainiM2009JNS} setting.

\smallskip

The flow rules proposed in Ref.~\cite{Gurti2004JMPS,GurtiA2005JMPS} are of particular interest because they incorporate two length scales:

\smallskip

\begin{itemize}
\item an energetic scale $L$, which appears from letting the free-energy density depend on derivatives of the \emph{plastic strain}, $\mathbf E^{\rm p}$, through the \emph{Burgers tensor}, $\mathbf G={\rm curl}\mathbf E^{\rm p}$;

\smallskip

\item a dissipative scale $\ell$, which arises from letting the gradient of plastic strain rate, $\nabla\dot{\mathbf E^{\rm p}}${,} enter the dissipation-rate density.
\end{itemize}

\smallskip

The form of the free energy density is motivated by dislocation mechanics. In particular, the choice of letting the free energy to depend on plastic strain gradient through the Burgers tensor follows from the presumption that the so--called geometrically-necessary dislocations (whose density is measured by $\mathbf G$) play a major role in determining size-dependent response, a presumption that finds its justification in homogenization results from discrete-dislocation models \cite{GarroLP2010JEMS,KratoS2008PRB}.

\smallskip

Because of the complicated nature of the non--local flow rule, it is not easy to understand how its solutions are affected by the material scales. On the other hand, such understanding is crucial in order to identify these scales by comparison with experiments. Thus, parallel with the literature dealing with modeling, researchers have also endeavored to investigate how the various scales may affect the nature of solutions, not only for the Gurtin-Anand theory, but also for other strain-gradient plasticity theories.

\smallskip

This task is usually accomplished by working out a simple analytical problem that mimics some experimental setup. For example, scale dependence for the torsion experiment was investigated in Ref.~\cite{IdiarF2010MSMS}
(by numerical and asymptotic considerations) in the framework of the Fleck \& Willis theory \cite{FleckW2009JMPS} and in Ref.~\cite{ChiriGT2012SJAM} (by rigorous arguments) for
energetic scale effects within the Gurtin-Anand theory \cite{GurtiA2005JMPS}. Moreover, for the distortion--gradient plasticity (which accounts also for plastic spin), specific finite-element schemes for the torsion problem have been recently proposed in Ref.~\cite{BardeP2014Modelling}. Problems involving {micro-}bending have been scrutinized in Ref.~\cite{IdiarDFW2009IJOES} and, more recently, in Ref.~\cite{FleckHW2014PRSAP} in the case of non-monotone loading.

\smallskip

With a similar goal in mind, a simplified flow rule, formulated in one spatial dimension, was derived and analyzed in Ref.~\cite{AnandGLG2005JMPS} to investigate the effects of both the energetic and the dissipative scales. Such flow rule, which mimics the traction problem in simple shear symmetry, will be introduced in Section \ref{S-model}. In the same section we also make a comparison with conventional plasticity. From our comparison two facts emerge: 1) that the length-scale $\ell$ is expected to be a source of additional strengthening; 2) that the natural way to quantify strengthening is to consider increase of the {\em Yield stress}, $\tau_Y$, i.e., the value of the (shear) stress that triggers plastic flow in an initially virgin sample.

\smallskip

The aim of this paper is to rigorously prove that, according to the flow rule we consider, $\tau_Y$ is strictly increasing with $\ell$, that is to say, \emph{smaller samples are stronger}. In fact, within this simplified framework of simple shear symmetry, we will be able to determine the dependence of $\tau_Y$ on $\ell$ explicitly. Results and proof are stated (in renormalized variables) in Section \ref{S-result}, which also contains an outline of the arguments (details are given in Sections \ref{S-proof-1}-\ref{S-proof-2}). In summary, we will first argue that $\tau_Y$ may characterized as the smallest value that the \emph{renormalized dissipation}
$$
\frac{S_0}{h} \int_{-h}^{+h} \sqrt{\phi^2(y)+\ell^2 \phi_y^2(y)}{\rm d} y
$$
attains among all $\phi\in H^1_0((-h,+h))$ such that $\int_{-h}^{+h}\phi(y){\rm d} y=1$ (see Theorem \ref{bbvstb}). This constrained minimization problem had already been introduced in \cite{AnandGLG2005JMPS} and analyzed in \cite{AmarCGR2013JMAA}, showing that a minimum is attained in $BV$, which is smooth in the interior and satisfies the corresponding Euler-Lagrange equation. We will then argue that these results permit to explicitly characterize $\tau_Y$ in terms of $\ell$ (see Theorem \ref{T2} and Figure \ref{fig:3}).

\section{Problem setup}\label{S-model}

\subsection{The traction problem}

The  one dimensional theory developed in Ref.~\cite{AnandGLG2005JMPS} describes plastic flow in a body  having the shape of an infinite strip of width $2h$, namely, $\Omega_h=\big\{\mathbf x=(x,y,z)\in\mathbb R^3:-h<y<h\big\}$, as sketched in Fig.~1.
\begin{figure}[h]
\centering
\includegraphics[scale=0.3]{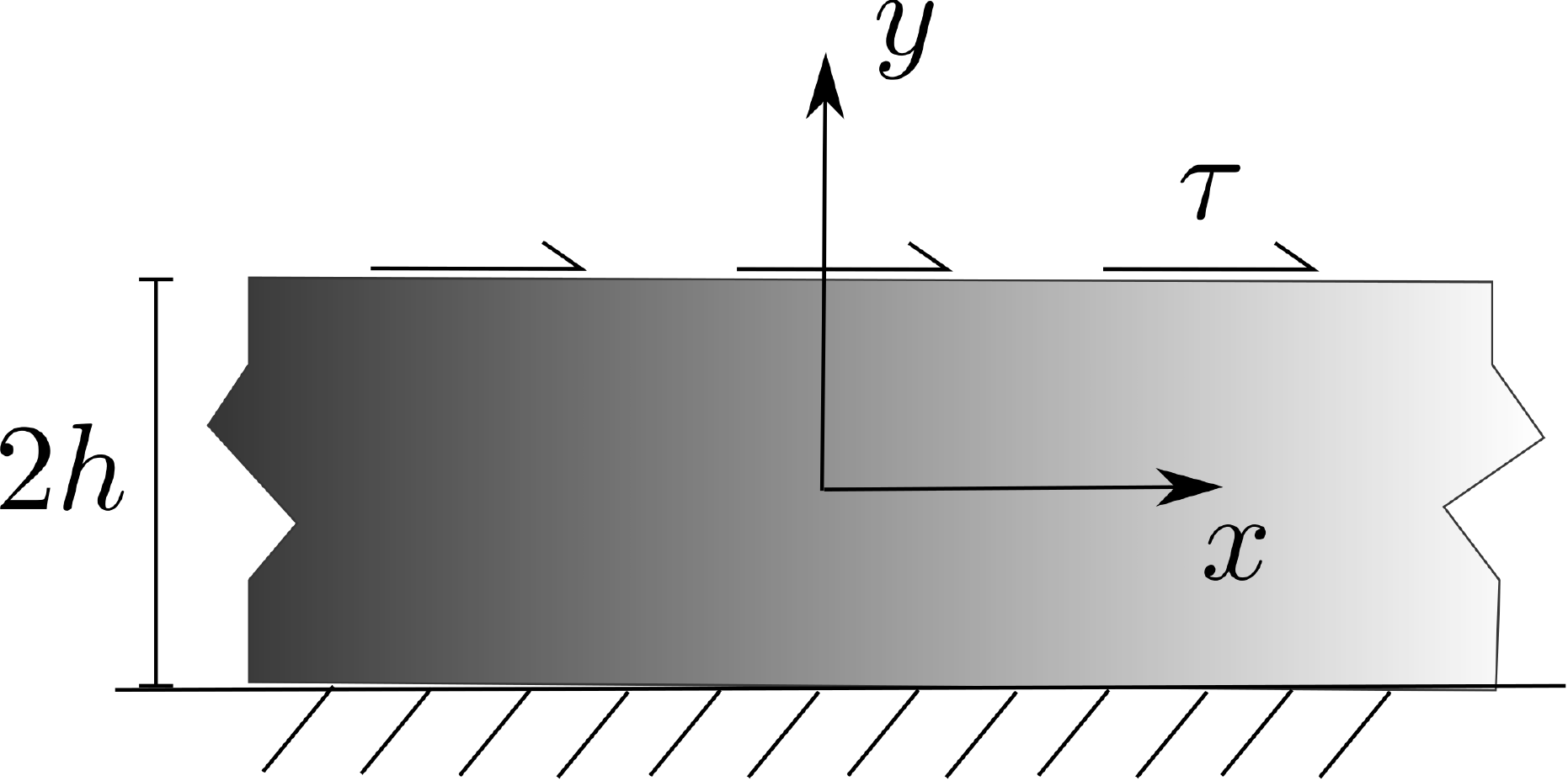}
\caption{An infinite strip of height $2h$, clamped on the bottom side and subject to a uniform shear traction $\tau$ on the top side.}
\end{figure}
We restrict attention to the so--called \emph{traction problem}, {describing} an ideal experiment in which the bottom side of the strip is clamped and a uniform \emph{shear stress} $\tau$ is prescribed on the upper side. We work in the \emph{rate-independent} setting of quasistatic evolution in plasticity and we limit our attention to the case of \emph{proportional loading}, that is to say, we assume that $\tau$ be strictly increasing with respect to time. With this assumption, we may label each instant by the corresponding value of the shear stress and adopt $\tau$ in place of time as the independent variable.

\smallskip

Because of translational invariance in the $x$- and $z$-directions, it is natural to look for solutions that enjoy the same invariance properties. Precisely, we assume that the two kinematic fields of interest, namely the \emph{displacement} $\mathbf u$ and the \emph{plastic strain} $\mathbf E^{\rm p}$, have the following representation:
\begin{equation*}
\begin{aligned}
%\label{eq:19}
&\mathbf u=u(y,\tau)\mathbf e_1,\quad \mathbf E^{\rm p}=\gamma^{\rm p}(y,\tau)\textrm{sym}(\mathbf e_1\otimes\mathbf e_2),
\end{aligned}
\end{equation*}
with $\{\mathbf e_i:i=1,2,3\}$ the canonical basis of $\mathbb R^3$. Consistent with this assumption, we take the stress tensor $\mathbf T$ to be spatially constant, and having the representation
$$
\mathbf T(\tau)=\tau(\mathbf e_1\otimes\mathbf e_2+\mathbf e_2\otimes\mathbf e_1).
$$

\subsection{A local flow rule: strengthening and hardening}\label{SS-local}

If the material is modeled in the framework of Mises plasticity with kinematic hardening, the flow rule governing the evolution of the \emph{shear strain} $\gamma^{\rm p}$ may be written as
\begin{equation}\label{eq:36}
\begin{cases}
 \tau-S_0\kappa\gamma^{\rm p}=\tau^{\rm dis},\\
 \frac{\tau^{\rm dis}}{S_0}\in
{\rm Sign}(\dot\gamma^{\rm p}),
\end{cases}
\end{equation}
where $S_0>0$ is the \emph{coarse-grain yield strength}, $\kappa$ is the \emph{kinematic-hardening coefficient}, a superimposed dot denotes differentiation with respect to the loading parameter $\tau$,
and
  \begin{equation*}
  {\rm Sign}(x)=\left\{\begin{array}{ll}
\{+1\} & \textrm{ if }x>0,
\\
{[-1,+1]} & \textrm{ if }x=0,
\\
\{-1\}  & \textrm{ if }x<0.
\end{array}\right.
\end{equation*}
Granted that the body is in its virgin state at the beginning of the experiment, namely,
\begin{equation}\label{virgin}
  \gamma^{\rm p}(y,0)=0,
\end{equation}
the solution of \eqref{eq:36} is easily worked out and, on introducing the positive-part operator $(\cdot){_+}{=\max\{\cdot, 0\}}$, can be written as
$$
\gamma^{\rm p}(y,\tau)=\frac{\left(\tau/{S_0}-1\right)_{_+}}{\kappa}.
$$
This solution displays the typical features of a stress-strain diagram from classical plasticity; in particular:

\smallskip

\begin{itemize}

\item the increase of $S_0$ is associated to \emph{strengthening}, that is, an increase of the threshold for the inset of plastic flow, the {\emph Yield shear stress}:
    \begin{equation}\label{def-Y}
\tau_Y=S_0;
    \end{equation}

\item the increase of $\kappa$, with $S_0$ fixed, is associated to \emph{hardening}, that is, an increase of the shear stress required to attain a given amount of plastic shear.
\end{itemize}

\smallskip

On multiplying \eqref{eq:36} by $\dot\gamma^{\rm p}$, we obtain the free energy balance
\begin{equation}\label{eq:42}
\frac 12 S_0\kappa\left({\frac{\rm d}{{\rm d}\tau}(\gamma^{\rm p})^2}\right)+{S_0|\dot\gamma^{\rm p}|}=\tau\dot\gamma^{\rm p},
\end{equation}
the \emph{free energy density} being given by $\frac {S_0}2\kappa(\gamma^{\rm p})^2$.
The balance \eqref{eq:42} can thus be interpreted as a splitting of the internal power $\tau\dot\gamma^{\rm p}$ expended on plastic flow into an \emph{energetic part} and a \emph{dissipative part}, $\tau^{\rm dis}\dot\gamma^{\rm p}=S_0|\dot\gamma^{\rm p}|$. Accordingly, we may say that, in the present context,

\smallskip

\begin{itemize}

\item strengthening is a \emph{dissipative effect}, whereas

\smallskip

\item hardening is an \emph{energetic effect}.

\end{itemize}

\subsection{A non-local flow rule: size-dependent strengthening and hardening}

Using the strain-gradient plasticity theory of Ref.~\cite{AnandGLG2005JMPS} we derive in Appendix 1 a non-local, rate-independent flow rule. In particular, we replace the first of \eqref{eq:36} with:
\begin{subequations}\label{eq:21}
\begin{align}
  \label{eq:21a}
  &\tau-S_0\left(\kappa \gamma^{\rm p}-L^2 \gamma^{\rm p}_{yy}\right)=\tau^{\rm dis}- k^{\rm dis}_y,
\end{align}
and the inclusion in \eqref{eq:36} with:
\begin{align}
  \label{eq:21b}
&\frac{(\tau^{\rm dis},{\ell^{-1}} k^{\rm dis})}{S_0}
\in
\textrm{Vers}\left(\dot\gamma^{\rm p},\ell\dot\gamma^{\rm p}_y\right),
\end{align}
\end{subequations}
where the index $y$ denotes partial differentiation with respect to $y$ and
\begin{equation*}
{\rm Vers}(\mathbf v) = \left\{\begin{array}{ll}
\left\{\frac{\mathbf v}{|\mathbf v|}\right\} & \quad\textrm{if}\quad|\mathbf v|\neq 0,\\
\{\mathbf v\in\mathbb R^2:\ |\mathbf v|\le 1{\}} & \quad\textrm{if}\quad |\mathbf v|=0.
\end{array}\right.
\end{equation*}
Problem \eqref{eq:21} must be complemented by both initial conditions, for which we again choose the virgin-state condition \eqref{virgin},
\begin{subequations}\label{eq:4}
\begin{equation} \label{eq:4a}
\gamma^{\rm p}|_{\tau=0}=0,
\end{equation}
and boundary conditions, for which we choose \emph{microscopic hard conditions}:
\begin{equation}\label{eq:4b}
\gamma^{\rm p}|_{y=-h}=\gamma^{\rm p}|_{y=+h}=0.
\end{equation}
\end{subequations}
As explained in Appendix 1, the partial differential equation \eqref{eq:21a} is {constitutively-augmented microforce balance} engendered by a version of the principle of virtual powers that accounts for power expenditure on the time derivative of the shear-strain gradient $\gamma^{\rm p}_y$. In particular,
taking the formal variation of the \emph{plastic free energy}
\begin{equation} \label{plastic-free-energy}
E^{\rm p}(\gamma^{\rm p})=\frac {S_0}2\int_{-h}^{+h}\left(\kappa(\gamma^{\rm p})^2+L^2\left(\gamma^{\rm p}_y\right)^2\right){\rm d} y
\end{equation}
and defining the \emph{plastic dissipation rate}
\begin{equation} \label{plastic-dissipation-rate}
\Psi^{\rm p}(\gamma^{\rm p})=S_0 \int_{-h}^{+h} \sqrt{(\dot{\gamma}^{\rm p})^2+\ell^2(\dot\gamma^{\rm p}_y)^2}{\rm d}y,
\end{equation}
the following identity is arrived at:
\begin{equation}\label{kl}
  \frac {\rm d}{{\rm d} \tau} E^{\rm p}(\gamma^{\rm p})+ \Psi^{\rm p}(\gamma^{\rm p}) = \int_{-h}^{+h}\tau\dot\gamma^{\rm p}{\rm d}y,
\end{equation}
which is again interpreted as a splitting of work expenditure (the right-hand side of \eqref{kl}) into an energetic part and a dissipative part. Given that $L$, resp. $\ell$, appear in the energetic, resp. dissipative part of the energy balance \eqref{kl}, in line with the discussion in \S \ref{SS-local}

\smallskip

\begin{itemize}
\item  one may expect that the extra energy brought into play by $L$ enhances hardening effects, and that the extra dissipation associated to $\ell$ is a source of additional strengthening.
\end{itemize}

\smallskip

The role of $L$ has been recently scrutinized in \cite{ChiriGT2012SJAM}, rigorously confirming this expectation in the case of torsion of thin wires. The role of $\ell$ has been investigated both formally and numerically in \cite{AnandGLG2005JMPS}. In view of the discussion in \S \ref{SS-local} (cf. in particular \eqref{def-Y}) a natural way to rigorously quantify the role of $\ell$ is to determine how the yield shear stress,
\begin{equation}\label{def-tau-Y}
  \tau_Y:=\sup\big\{\tau\ge 0:\gamma^{\rm p}(\cdot,\tau)=0\big\},
\end{equation}
depends on $\ell$. This constitute the goal of this paper.

\section{Main results}\label{S-result}
\subsection{Scaling}
In order to single out the relevant parameters, we introduce dimensionless independent variables:
\begin{equation*}%\label{normalized}
r:=\frac {y} h, \quad\theta:=\frac{\tau}{S_0}.
\end{equation*}
Consistent with this choice, we introduce the dimensionless parameters:
\begin{equation}\label{eq:9}
\lambda:=\frac {\ell} {h},\quad\Lambda:=\frac{L}{h}.
\end{equation}
The nonlocal flow rule \eqref{eq:21} can now be reformulated in the domain $I:=(-1,+1)$ and takes the form (henceforth, for typographical convenience, we drop the superscript $\rm p$ from $\gamma^{\rm p}$):
\begin{equation}
  \label{eq:34}
  \begin{cases}
  \displaystyle\theta-\kappa\gamma+\Lambda^2\gamma_{rr}=\bar\tau^{\rm dis}-\bar k^{\rm dis}_r,\\
\displaystyle (\bar\tau^{\rm dis},\bar k^{\rm dis})\in \partial \psi_\lambda\left(\dot\gamma,\dot\gamma_r\right),
  \end{cases}
\end{equation}
where the index $r$ denotes partial differentiation with respect to $r$. Initial and microscopically hard boundary conditions \eqref{eq:4} now read as
\begin{equation}\label{eq:2}
  \gamma(r,0)=\gamma(-1,\theta)=\gamma(+1,\theta)=0\quad(r,\theta)\in I\times[0,+\infty)
\end{equation}
and the \emph{renormalized plastic free energy}, resp. \emph{dissipation-rate}, are given by
\begin{equation} \label{eq:g1}
E(\gamma):=  \frac{\kappa}2\int_I\left(\gamma^2+\Lambda^2\gamma_r^2\right){\rm d} r, \quad   \Psi(\gamma):=\int_I\sqrt{\gamma^2+\lambda^2\gamma_r^2}{\rm d}r
\end{equation}
(cf. \eqref{plastic-free-energy}, resp. \eqref{plastic-dissipation-rate}). In renormalized varibales, our aim becomes that of rigorously quantifying the dependence on the {\emph renormalized dissipative scale}, $\lambda$, of the {\emph renormalized yield shear stress} (cf. \eqref{def-tau-Y})
\begin{equation}
  \label{eq:48}
 \frac{\tau_y}{S_0}= \theta_Y:=\sup\big\{\theta\ge 0:\gamma(\vartheta)=0\quad \forall \vartheta \in [0,\theta]\big\},
\end{equation}
namely, \emph{the largest value attained by the renormalized shear stress $\theta$ prior to the inset of plastic flow.}

\subsection{Energetic formulation}

We assume hereafter that $\kappa\ge 0$, $\Lambda>0$, and $\lambda>0.$ Being a \emph{rate-independent} dynamical system, the flow rule \eqref{eq:34}--\eqref{eq:2} can be formulated in many equivalent ways. The formulation that best suits our needs is the so--called \emph{energetic formulation} proposed in Ref.~\cite{MielkT2004NNDEA}. With a view towards formulating \eqref{eq:34}--\eqref{eq:2} in the energetic format, we introduce the (renormalized) \emph{energy functional}:
\begin{equation}
  \label{eq:20}
  \mathscr E(\theta,\gamma):=E(\gamma)-\theta\int_I \gamma{\rm d}r.
\end{equation}
As usual, we write $\gamma(\theta):=\gamma(\theta,\cdot)$. We can now give the definition of energetic solution.

\smallskip

\begin{definition}[Energetic solution]\label{def-sol}
 Given $\Theta>0$, a function $\gamma:[0,\Theta]\to H^1_0(I)$ is an \emph{energetic solution} to \eqref{eq:34}--\eqref{eq:2} if the function ${[0,\Theta]\ni }\theta\mapsto\frac{\partial\mathscr E}{\partial\theta}(\theta,\gamma(\theta))=-\int_I\gamma{\rm d}r$ is in $L^1((0,\Theta))$ and if the following conditions are satisfied for all $\theta\in[0,\Theta]$:
\begin{subequations}\label{eq:14}
    \begin{align}
      &\mathscr E(\theta,\gamma(\theta))\le \mathscr E(\theta,v)+\Psi(\gamma(\theta)-v) \quad\mbox{ for all $v\in H^1_0(I)$}, \label{eq:50}\\
      &\mathscr E(\theta,\gamma(\theta))+{\rm dis}_{\Psi}(\gamma;[0,\theta])=-\int_0^{\theta}\int_I \gamma(\vartheta){\rm d}r\,{\rm d}\vartheta,\label{eq:30}
    \end{align}
\end{subequations}
where ${\rm dis}_{\Psi}(\gamma;[0,\theta])$ is the total variation of $\gamma$ on $[0,\theta]$ with respect to the distance $d(\gamma_1,\gamma_2)=\Psi(\gamma_1-\gamma_2)$, i.e.,
$$
{\rm dis}_{\Psi}(\gamma;[0,\theta]):=\sup\left\{\sum_{j=1}^N \Psi(\gamma(\theta_j)-\gamma(\theta_{j-1})):\ N\in \N,\ 0= \theta_0<\dots <\theta_N=\theta \right\}.
$$
\end{definition}

\smallskip

In the present setting (quadratic energy) the next proposition is established without burden by invoking, for instance, Theorem 2.1 in Ref.~\cite{Mielk2005Evolution}:

\smallskip

\begin{proposition}
  There exists a unique solution $\gamma$ of \eqref{eq:34}--\eqref{eq:2}. Moreover, $\theta\mapsto\gamma(\theta)$ is Lipschitz continuous as a function from $[0,\Theta]$ to $H^1_0(I)$.
\end{proposition}

\subsection{Characterizations of $\tau_Y$}\label{sec:actu-yield-strength}

The first main result of this paper is the following characterization of $\theta_Y$:

\smallskip

\begin{theorem}\label{bbvstb}
Let $\gamma$ be the unique energetic solution to \eqref{eq:34}-\eqref{eq:2} and let $\theta_Y$ as in \eqref{eq:48}. Then
\begin{equation}
  \label{eq:15}
  \theta_Y=\inf\left\{\Psi(\phi)
{: \ \phi}\in H^1_0(I),\ \int_I{\phi}{\rm d}r=1\right\}.
\end{equation}
\end{theorem}

\smallskip

In order to explain the relation between the two quantities, it is convenient to briefly illustrate the main steps in the proof, whose details are given in \S \ref{S-proof-1}.
We begin by observing that the energy--balance condition \eqref{eq:30} is identically satisfied for all $\theta\in(0,\theta_Y)$. Thus, what determines the inset of plastic flow is the loss of stability of the trivial state $\gamma\equiv 0$. This leads us to consider the \emph{stability indicator}:
\begin{equation}
  \label{eq:16}
 m(\theta):=\inf_{{\phi}\in H^1_0{(I)}}\Phi_\theta({\phi}), \quad\mbox{where}\quad
  \Phi_\theta({\phi}):=\mathscr E(\theta,{\phi})+\Psi({\phi}).
\end{equation}
We will indeed argue that
\begin{equation*}
 \theta_Y=\inf\left\{\theta\ge 0: \ m(\theta)<0\right\}
\end{equation*}
(cf. Proposition \ref{*vsb}). Next, we note that the plastic dissipation rate $\Psi$ is (positively) homogeneous of degree one in $\gamma$, whereas the plastic free energy $E$ is quadratic. Then, a simple scaling argument can be used to show that the \emph{reduced stability indicator}
\begin{equation}
  \label{eq:27}
  \widetilde m(\theta):=\inf_{{\phi}\in H^1_0{(I)}}\widetilde \Phi_{\theta}({\phi}), \quad\mbox{where}\quad \widetilde \Phi_{\theta}({\phi}):= \Psi(\phi)  -\theta\int_I {\phi}\,{\rm d}r
\end{equation}
is equivalent to the \emph{stability indicator}:
\begin{equation*}%\label{eq:25}
  m(\tau)<0 \Leftrightarrow \widetilde m(\tau)<0
\end{equation*}
(cf. Proposition \ref{bvsbb}). The last step of our argument consists in observing that, again by homogeneity,  for negative values of $\widetilde m$ we can restrict our attention to the subspace of tests $\phi$ satisfying the normalization condition $\int\limits_I \phi{\rm d}r =1$: this leads to Theorem \ref{bbvstb}.

\subsection{The formula for $\tau_Y$}
The second main result of this paper is the following explicit formula for $\tau_Y$:

\smallskip

\begin{theorem}\label{T2}
The renormalized yield shear stress $\theta_Y=\displaystyle\frac{\tau_Y}{S_0}$ and the \emph{renormalized diassipative scale} $\lambda=\displaystyle\frac \ell h$ are related by
\begin{equation}\label{eq:65}
    \lambda =
\frac{2\sqrt{\theta_{Y}^2-1}}{\pi(\theta_{Y}
-\sqrt{\theta_{Y}^2-1})+2\theta_{Y}\arctan
\frac 1 {\sqrt{\theta_{Y}^2-1}}}.
\end{equation}
\end{theorem}

\smallskip

The proof is provided in Section \ref{S-proof-2} and relies on results from Ref.~\cite{AmarCGR2013JMAA}, guaranteeing that the relaxation in $BV(I)$ of the infimum problem in \eqref{eq:15} admits a minimizer $\phi_Y$ which is smooth in $I$ and satisfies the Euler-Lagrange equation
\begin{equation}
    \label{eq:53}
    \theta_Y=\frac{{\phi_Y}}{\sqrt{{\phi_Y}^2+\lambda^2\big(\frac{{\rm d}\phi_Y}{{\rm d} r}\big)^2}}-\lambda^2
\frac{{\rm d}}{{\rm d} r}\frac{\frac{{\rm d}\phi_Y}{{\rm d} r}}{\sqrt{{\phi_Y}^2+\lambda^2\big(\frac{{\rm d}\phi_Y}{{\rm d} r}\big)^2}}.
\end{equation}
%
%second, that the minimizer is even, positive, strictly decreasing for $y\ge 0$ and its derivative blows up at the boundary.
%
By a suitable change of dependent variable, we converts \eqref{eq:53} into a \emph{first-order} differential equation with \emph{two} side conditions. The extra side condition selects the \emph{eigenvalue} $\theta_Y$ of the E-L equation \eqref{eq:53}, yielding \eqref{eq:53}.

\smallskip

The graph of $\tau_Y/S_0$, recovered from \eqref{eq:65}, is plotted in Fig. \ref{fig:3} (recall \eqref{eq:9} and \eqref{eq:48}). Our result confirms that as the sample becomes smaller{, i.e. $\lambda =\ell/h$ increases}, the actual yield strength increases{:} hence \emph{smaller samples are stronger}. Needless to say, the results from out plot agree with the numerical calculations carried out in Ref.~\cite{AnandGLG2005JMPS} and reported in Figure 4 thereof.
\begin{figure}[h]
  \centering
%imported with text as graphics
\def\svgwidth{0.7\linewidth}
%\input{plot_tau_vs_ell.pdf_tex}
%% Creator: Inkscape inkscape 0.48.4, www.inkscape.org
%% Accompanies image file 'plot_tau_vs_ell.pdf' (pdf, eps, ps)
\begingroup%
  \makeatletter%
  \providecommand\color[2][]{%
    \errmessage{(Inkscape) Color is used for the text in Inkscape, but the package 'color.sty' is not loaded}%
    \renewcommand\color[2][]{}%
  }%
  \providecommand\transparent[1]{%
    \errmessage{(Inkscape) Transparency is used (non-zero) for the text in Inkscape, but the package 'transparent.sty' is not loaded}%
    \renewcommand\transparent[1]{}%
  }%
  \providecommand\rotatebox[2]{#2}%
  \ifx\svgwidth\undefined%
    \setlength{\unitlength}{436.55445923bp}%
    \ifx\svgscale\undefined%
      \relax%
    \else%
      \setlength{\unitlength}{\unitlength * \real{\svgscale}}%
    \fi%
  \else%
    \setlength{\unitlength}{\svgwidth}%
  \fi%
  \global\let\svgwidth\undefined%
  \global\let\svgscale\undefined%
  \makeatother%
\null
\vskip 5em
\null
\begin{centering}
  \begin{picture}(0.7,0.41355894)%
    \put(0,0){\includegraphics[width=0.7\unitlength]{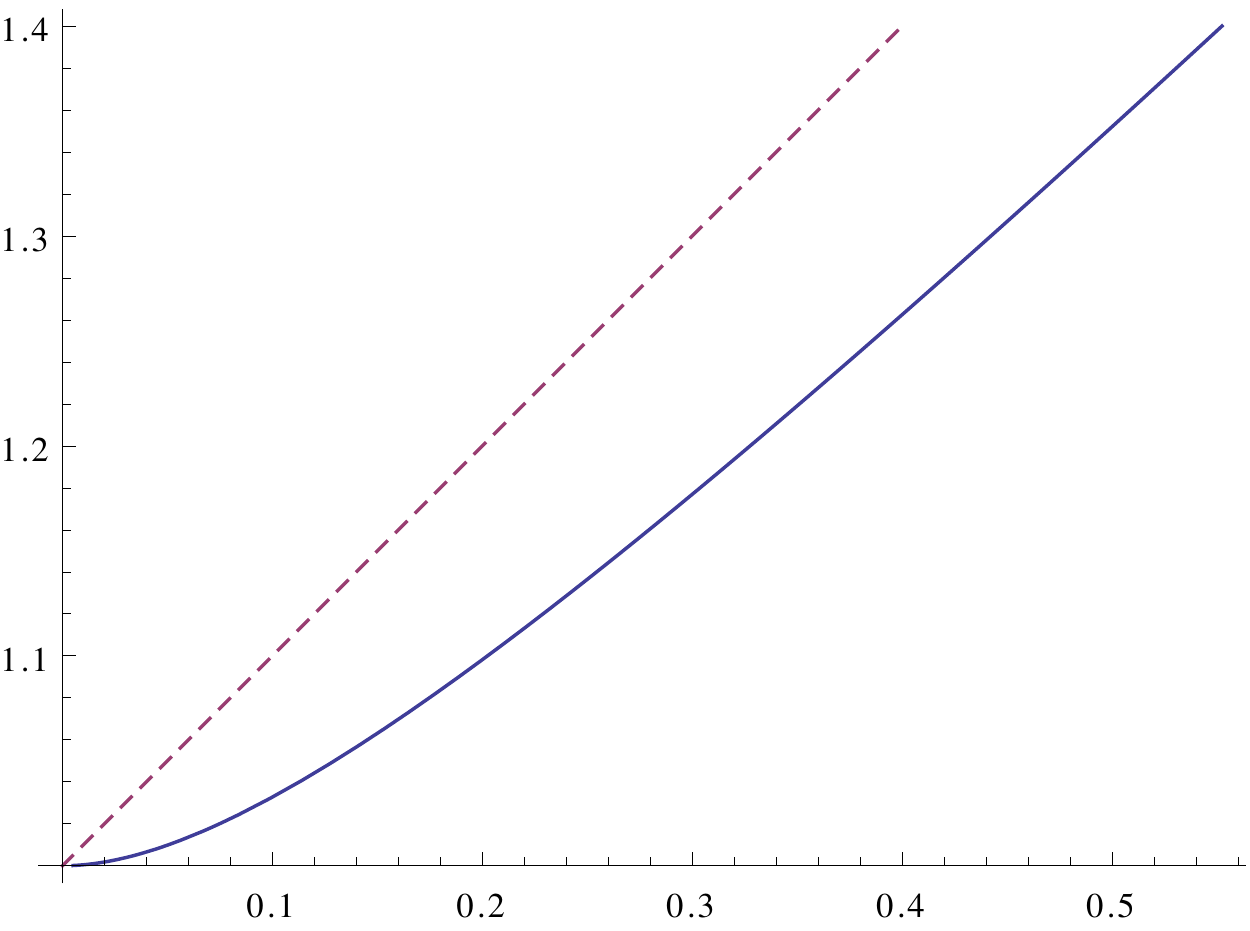}}%
    \put(0.6148828,0.0715254){\color[rgb]{0,0,0}\makebox(0,0)[lb]{\smash{$\color{black}\lambda=\frac \ell h$}}}%
    \put(0.07499698,0.48209667){\color[rgb]{0,0,0}\makebox(0,0)[lb]{\smash{$\color{black}\theta_Y=\frac{\tau_Y}{S_0}$}}}%
  \end{picture}%
\end{centering}
\endgroup%
  \caption{Solid line: renormalized effective yield strength $\tau_Y/S_0$ versus renormalized dissipative scale $\ell/h$, as from formula \eqref{eq:65}. Dashed line: the upper bound $\frac{\tau_Y} {S_0}<1+\frac {\ell}h$ derived in [5]. This plot agrees with the result computed numerically in [5] and reported in Fig. 4 thereof. When comparing the two figures, the reader should take into account that in the present paper the symbol $h$ denotes half the thickness of the strip, whereas in [5] the same symbol denotes the overall thickness.
}
  \label{fig:3}
\end{figure}

Our explicit formula provides additional insight concerning the asymptotic behavior of the actual yield strength for small and large values of $h$. In particular, from \eqref{eq:65} one finds that, for $0<\theta_Y-1\ll 1$,
\begin{equation*}
  \lambda  \sim \frac{\sqrt 2} \pi \sqrt{{\theta_Y}-1},
\end{equation*} which implies that, for $0<\lambda\ll 1$, the renormalized actual yield strength has the following asymptotic behavior:
\begin{equation*}
  \theta_Y -1 \sim \frac {\pi^2} 2\lambda^2 {\quad\mbox{for }\ 0<\lambda\ll 1}.
\end{equation*}
We also note that, as $\lambda=\ell/h \to \infty$, a linear relation is recovered:
$$
\theta_{{Y}} -\lambda \sim \frac{\pi}{4}  \quad\mbox{for } \ \lambda\gg 1.
$$

\section{Proof of Theorem \ref{bbvstb}}\label{S-proof-1}

Existence and uniqueness of the minimum in \eqref{eq:16} is readily ascertained through the direct method of the calculus of variations, owing to coercivity, lower semicontinuity, and convexity of $\Phi_\theta$ in $H^1_0(I)$:

\smallskip

\begin{lemma}
For any $\lambda>0$ there exists a unique minimizer $\phi_\lambda$ of the infimum problem in \eqref{eq:16}.
\end{lemma}

\smallskip

The first step is to show that if the trivial state is not stable at a certain value of the renormalized shear stress $\theta$ during the loading process, then it is not stable for whatever higher value:

\smallskip

\begin{lemma}\label{lem:8}
The function ${[0,\Theta]\ni }\theta\mapsto m(\theta)$ defined in \eqref{eq:16} is non-increasing.
\end{lemma}

\smallskip

\begin{proof}
Let ${\phi}_{\theta}$ be the unique minimizer of $\Phi_\theta$. First, we observe that
\begin{equation}\label{eq:g2}
{\phi}_{\theta}{\ge 0} \quad\mbox{ a.e. in $I$ {for all $\theta\ge 0$}.}
\end{equation}
Indeed, obviously $\phi_0\equiv 0$; for $\theta>0$, if $\phi_\theta<0$ in a set $J$ of positive measure, then {(by the definitions \eqref{eq:20} and \eqref{eq:g1} of $\mathcal E$, resp. $\Psi$)} we would have $\Phi_\theta(|\phi_\theta|)<\Phi_\theta(\phi_\theta)$, a contradiction. Thus, given $\theta_1\le\theta_2$, we have
\begin{eqnarray}
  \nonumber
  m(\theta_2) &\stackrel{\eqref{eq:16}}=& \Phi_{\theta_2}({\phi}_{\theta_2})
  \\ \nonumber &\le& \Phi_{\theta_2}({\phi}_{\theta_1}) \quad\mbox{(by def. of ${\phi}_{\theta_2}$)}
  \\
  &\stackrel{\eqref{eq:20},{\eqref{eq:g2}}}\le& \Phi_{\theta_1}({\phi}_{\theta_1}) {\stackrel{\eqref{eq:16}}=} m(\theta_1),
  \nonumber %\label{eq:1}
\end{eqnarray}
as desired.
\end{proof}

\smallskip

The previous lemma is expedient to arrive to the following characterization of $\theta_Y$.

\smallskip

\begin{proposition}\label{*vsb}
Let $\gamma$ be the unique energetic solution to \eqref{eq:34}-\eqref{eq:2} and let $\theta_Y$ and $m$ as in \eqref{eq:48}, resp. \eqref{eq:16}. Then
$$
\theta_Y=\inf\left\{\theta\ge 0:\ m(\theta)<0\right\}.
$$
\end{proposition}

\smallskip

\begin{proof} Let us set $\widehat\theta=\inf\left\{\theta\ge 0:\ m(\theta)<0\right\}$. We notice that, since $m(\theta)$ is nonincreasing, $m(\theta)=0$ in $[0,\widehat\theta)$. Hence, by direct substitution into \eqref{eq:14}, we see that the trivial function $\theta\mapsto 0$ is an energetic solution on the interval $[0,\widehat\theta)$. By the uniqueness of the energetic solution, and by \eqref{eq:48}, it follows that $\theta_Y\ge \widehat\theta$.

\smallskip

The reverse inequality follows from the monotonicity of $\theta\mapsto m(\theta)$: suppose indeed that $\widehat\theta < \theta_Y$; then, by Lemma \ref{lem:8} there exists $\tilde\theta<\theta_Y$ such that $m(\tilde\theta)<0$; however, $\tilde\theta<\theta_Y$ implies that $\gamma(\tilde\theta)=0$; thus, by \eqref{eq:50} and \eqref{eq:16}, this means that $m(\tilde \theta)=0$, whence a contradiction.
\end{proof}

\smallskip

We now show that the reduced stability indicator defined in \eqref{eq:27} can be used to detect the inset of plastic flow. Indeed, we have the following equivalence:

\smallskip

\begin{proposition}\label{bvsbb}
\begin{equation}\label{eq:8}
\theta_Y=\inf\{\theta\ge 0:\widetilde m(\theta)<0\}.
\end{equation}
\end{proposition}

\smallskip

\begin{proof}
In view of Proposition \ref{*vsb}, it suffices to show that
\begin{equation*}
%    \label{eq:37}
   m(\theta)<0\qquad\textrm{if and only if}\qquad \widetilde m(\theta)<0.
  \end{equation*}
Since by definition $\tilde \Phi_\theta \le \Phi_\theta$, $m(\theta)<0$ obviously implies $\tilde m(\theta)<0$.
For the reverse implication, let us assume  $\widetilde m(\theta)<0$. Then there exists $\widetilde{\phi}\in H^1_0(I)$ such that $\widetilde\Phi_\theta(\widetilde{\phi})<0$. On the other hand, by the 1-homogeneity of $\widetilde\Phi_\theta$,
\begin{equation*}
%  \label{eq:40}
  \lim_{\alpha\to 0+} \frac{\Phi_\theta(\alpha\widetilde{\phi})}\alpha=\widetilde\Phi_\theta(\widetilde{\phi})<0.
\end{equation*}
Thus ${\Phi_\theta(\alpha\widetilde{\phi})}<0$ for $\alpha>0$ sufficiently small, whence $m(\theta)<0$.
\end{proof}

\smallskip

With Proposition \ref{bvsbb} at hand we are now ready to establish the variational characterization we have been after.

\smallskip

\begin{proof}[Proof of Theorem \ref{bbvstb}]
Let
\begin{equation}\label{eq:11}
\widehat\theta_Y(\lambda):=\inf\left\{\Psi(\phi)
{: \ \phi}\in H^1_0(I),\ \int_I{\phi}{\rm d}r=1\right\}.
\end{equation}
On recalling the definitions of $\Psi$ and $\widetilde\Phi_\theta$ given in \eqref{eq:g1},  respectively \eqref{eq:27}, we see that the inequality
\begin{equation}\label{eq:32}
 \theta_Y\le\widehat\theta_Y(\lambda)
\end{equation}
is implied by the following chain of implications:
\begin{equation*}
%  \label{eq:17}
  \begin{aligned}
    \widehat\theta_Y(\lambda)<\theta
\quad
&\stackrel{\phantom{(0.0)}}\Rightarrow \quad\Psi(\bar{\phi})<\theta
\textrm{ for some $\bar{\phi}\in H^1_0(I)$
{such that} }
\int_I\bar{\phi}{\rm d}r=1\\
&\stackrel{\phantom{(0.0)}}\Rightarrow
\quad
 \inf_{\phi\in H^1_0(I)}\left(-\int_I \theta{\phi} {\rm d}r +\Psi({\phi})\right)<0\\
&\stackrel{\eqref{eq:27}}\Rightarrow \quad \widetilde m(\theta)
<0
\\
&\stackrel{\eqref{eq:8}}\Rightarrow
\quad
\theta_Y\le\theta.
  \end{aligned}
\end{equation*}
Having established \eqref{eq:32}, it remains for us to prove the reverse inequality:
\begin{equation}\label{eq:33}
  \theta_Y\ge \widehat\theta_Y(\lambda).
\end{equation}
To this aim, let $\theta\in (0, \widehat\theta_Y(\lambda))$. By the definition \eqref{eq:11} of $\widehat\theta_Y(\lambda)$, we have
\begin{equation}\label{ui1}
\theta\int_I \phi{\rm d} r=\theta <\Psi({\phi})\quad\mbox{for all ${\phi}\in H^1_0(I)$
such that $\displaystyle \int_I{\phi}{\rm d}r=1$}.
\end{equation}
Since both sides of the inequality in \eqref{ui1} are positively $1$-homogeneous, \eqref{ui1} upgrades to
\begin{equation}\label{ui2}
\theta\int_I \phi{\rm d} r <\Psi({\phi})\quad\mbox{for all ${\phi}\in H^1_0(I)$
such that $\displaystyle \int_I{\phi}{\rm d}r>0$}.
\end{equation}
In turn, since $\Psi$ is non-negative, \eqref{ui2} upgrades to
\begin{equation}\label{ui3}
0\le \Psi({\phi}) - \theta\int_I \phi{\rm d} r \stackrel{\eqref{eq:27}}= \widetilde \Phi_\theta({\phi}) \quad\mbox{for all ${\phi}\in H^1_0(I)$}
\end{equation}
which holds for all $\theta\in (0, \widehat\theta_Y(\lambda))$.
Summing up, we have the implication:
$$
0\le\theta<\widehat\theta_Y(\lambda) \ \stackrel{\eqref{ui3}}\Rightarrow \ \widetilde
m(\theta)\ =\  \inf_{\phi\in H^1_0
(I)} \widetilde\Phi_\theta({\phi}) \ge 0 \
\stackrel{\eqref{eq:8}}\Rightarrow \ \theta\le\theta_Y,
$$
whence \eqref{eq:33}, since $\theta_Y\ge 0$ by definition.
\end{proof}

\section{Proof of Theorem \ref{T2}}\label{S-proof-2}%\label{sec:comp-actu-yield}

The infimum problem in \eqref{eq:15} was addressed in Ref.~\cite{AmarCGR2013JMAA}. Consider the {\it relaxation} of $\Psi$,
\begin{equation}
    \label{eq:59}
    \bar\Psi({\phi}):=\inf\Big\{\liminf_{k\to\infty}\Psi({\phi}_k):\ \{{\phi}_k\}\subseteq W^{1,1}_0(I), \ {\phi}_k\to{\phi}\textrm{ in }L^1(I)\Big\},
\end{equation}
i.e. the largest lower semicontinous extension of $\Psi$.
It is shown in Ref.~\cite{AmarCGR2013JMAA} that the
relaxation $\bar\Psi$ has the following representation for $\phi\in BV(I)$:\footnote{Here $\|\mu\|$ denotes the total variation of a measure $\mu$ (see e.g. \cite[Def. 1.4]{AmbroFP2000}) and $\frac{{\rm d} \phi}{{\rm d} r}$, resp. $D^s \phi$, denote the absolutely continuous, resp. singular, part  of $D\phi$ with respect to the Lebesgue measure (see e.g. \cite[Th. 1.28 and \S 3.9]{AmbroFP2000}). We also refer to \cite{AmbroFP2000} for definitions and basic properties of the spaces $BV(I)$ and $SBV(I)$.}
\begin{equation}
  \label{eq:60}
\bar\Psi (\phi)= \int_I \sqrt{{\phi}^2+\lambda^2\left(\frac{{\rm d}\phi}{{\rm d}r}\right)^2}{\rm d}r+\lambda \|D^s{\phi}\|(I)+\lambda\left(|\phi(-1)| +|\phi(+1)|\right).
\end{equation}
Notice that, as is customary in the $BV$ setting, homogeneous boundary conditions are now incorporated in the functional through the penalization term $\lambda|\phi|(\partial I)=\lambda|\phi-0|(\partial I)$, which measures the jump between the trace of $\phi$ and the prescribed null value.

\smallskip

The following results were established in Ref.~\cite{AmarCGR2013JMAA}.

\smallskip

\begin{theorem}[see Thm. 5.1 in Ref.~\cite{AmarCGR2013JMAA}] \label{T51}
Let $\bar\Psi$ as in \eqref{eq:59}. There exists a unique ${\phi}_Y\in SBV(I)$ such that $\int_I{\phi}_Y{\rm d}r=1$ and
\begin{equation*}
%  \label{eq:43}
  \bar\Psi({\phi_Y})=\min\Big\{\bar\Psi({\phi}):{\phi}\in L^1(I), \int_I{\phi}{\rm d}r=1\Big\}.
\end{equation*}
Moreover, ${\phi}_Y$ is even, strictly decreasing in $[0,1)$, and smooth in $(-1,1)$; furthermore, it  solves the \emph{Euler-Lagrange equation}
\begin{equation}
    \label{eq:53bis}
\bar\Psi({\phi_Y}) =\frac{{\phi_Y}}{\sqrt{{\phi_Y}^2+\lambda^2\big(\frac{{\rm d}\phi_Y}{{\rm d} r}\big)^2}}-\lambda^2\frac{{\rm d}}{{\rm d} r}\frac{\frac{{\rm d}\phi_Y}{{\rm d} r}}{\sqrt{{\phi_Y}^2+\lambda^2\big(\frac{{\rm d}\phi_Y}{{\rm d} r}\big)^2}} \quad\mbox{in $I$}
\end{equation}
and it satisfies
\begin{equation}
  \label{eq:56bis}
\lim_{r\to 1^-}\frac{{\phi}_Y(r)}{{\phi}_Y(0)}=\frac{\theta_Y-1}{\theta_Y}\quad\mbox{and}\quad \lim_{r\to 1^-} \frac{{\rm d}\phi_Y}{{\rm d} r}(r)=-\infty.
\end{equation}
\end{theorem}

\smallskip

\begin{remark}{\rm
Notably, \eqref{eq:56bis} shows that the solution $\phi_Y\in SBV(I)$ of the relaxed minimization problem
does {\em not} satisfy the boundary conditions ${\phi}(-1)={\phi}(1)=0$;
generally speaking, this amounts to saying that, in order to minimize $\bar\Psi$ with mass constraint, paying a jump discontinuity at the boundary is cheaper than attaining the boundary value zero.
}\end{remark}

\smallskip

We are now ready to prove Theorem \ref{T2}.

\smallskip

\begin{proof}[Proof of Theorem \ref{T2}]
In view of Theorem \ref{bbvstb} and since $H^1_0(I)$ is dense in $BV(I)$,
\begin{equation}\label{passo2}
\bar\Psi({\phi_Y})=\theta_Y.
\end{equation}
We also notice that, since ${{\rm d}{\phi}_Y}/{{\rm d}r}<0$ in $[0,1)$ and $\phi_Y$ is positive with $\int_I\phi_Y(r){\rm d} r=1$,
\begin{equation}\label{eq:31}
\theta_Y \stackrel{\eqref{passo2}}= \bar\Psi({\phi_Y}) \stackrel{\eqref{eq:60}} \ge  \int_I \sqrt{{\phi_Y}^2+\lambda^2\Big(\frac{{\rm d}\phi_Y}{{\rm d}r}\Big)^2}{\rm d}r > \int_I |\phi_Y| {\rm}d r =1.
\end{equation}
Now, consider the function
\begin{equation*}%\label{eq:28}
\zeta(r):=-\lambda\frac{\frac{{\rm d}\phi_{Y}}{{\rm d} r}}{\sqrt{\phi_{Y}^2 +\lambda^2\left(\frac{{\rm d}\phi_{Y}}{{\rm d} r}\right)^2}}{.}
\end{equation*}
Since $\phi_{Y}$ is smooth and positive in $I$, $\zeta$ is smooth as well. We note that
\begin{eqnarray*}
1-\zeta^2 = 1- \lambda^2 \frac{\left(\frac{{\rm d}\phi_{Y}}{{\rm d} r}\right)^2}{\phi_{Y}^2 +\lambda^2\left(\frac{{\rm d}\phi_{Y}}{{\rm d} r}\right)^2}.
= \frac{\left(\phi_{Y}\right)^2}{\phi_{Y}^2 +\lambda^2\left(\frac{{\rm d}\phi_{Y}}{{\rm d} r}\right)^2}.
\end{eqnarray*}
Hence, since $\phi_Y>0$,
\begin{equation}\label{passo1}
\frac{\phi_{Y}}{\sqrt{\phi_{Y}^2 +\lambda^2\left(\frac{{\rm d}\phi_{Y}}{{\rm d} r}\right)^2}}= \sqrt{1-\zeta^2}.
\end{equation}
By making also use of the Euler-Lagrange equation, we see that $\zeta$ satisfies the following differential equation:
\begin{eqnarray}%\nonumber
  \lambda\frac{{\rm d}\zeta}{{\rm d}r}
   \stackrel{\eqref{eq:53bis}}=  \bar\Psi({\phi_Y})  - \frac{\phi_Y}{\sqrt{\phi_{Y}^2 +\lambda^2\left(\frac{{\rm d}\phi_{Y}}{{\rm d} r}\right)^2}}
   \label{eq:26} & \stackrel{\eqref{passo1},\eqref{passo2}}= &
  \theta_Y-\sqrt{1-\zeta^2}.
\end{eqnarray}
It follows from \eqref{eq:31} and \eqref{eq:26} that $\frac{\rm d\zeta}{{\rm d}r}>0$. Hence,
\begin{equation*}%\label{eq:23bis2}
\frac{1}{\lambda} \stackrel{\eqref{eq:26}}=\int_{0}^{1} \frac{\frac{{\rm d}\zeta}{{\rm d}r}}{\theta_Y-\sqrt{1-\zeta^2}}\,{\rm d}r = \int_{\zeta(0)}^{\zeta(1-)} \frac 1 {\theta_Y-\sqrt{1-\zeta^2}} {\rm d}\zeta.
\end{equation*}
In addition, since $\phi_{Y}$ is even and because of \eqref{eq:56bis}, we have that
\begin{equation*}%\label{passo3}
\zeta(0)=0\quad\mbox{and}\quad  \lim_{r\to 1-}\zeta(r)=1.
\end{equation*}
Therefore,
\begin{equation}\label{eq:23bis}
\frac{1}{\lambda} = \int_{0}^{1} \frac{{\rm d}\zeta}{\theta_Y-\sqrt{1-\zeta^2}}.
\end{equation}
The integral on the right-hand side of \eqref{eq:23bis} is well defined and can be computed explicitly. As a result, we arrive at formula \eqref{eq:65} for the renormalized actual yield stress.
\end{proof}

\appendix

\section{The nonlocal flow rule}

In this section we briefly recapitulate the steps leading to the flow rule \eqref{eq:21}, as devised in Ref.~\cite{AnandGLG2005JMPS}, with a few changes from the original path. At variance with the previous sections, we do not assume proportional loading. Accordingly, the independent variables are now $y\in (-h,+h)$ and $t$, which stands for time, and the index $y$ denotes partial differentiation with respect to $y$.

\subsection{Principle of virtual powers}

We start from the decomposition
\begin{equation}\label{eq:5}
  u_y=\gamma^{\rm e}+\gamma^{\rm p}
\end{equation}
of the \emph{shear strain} $u_y$ into an \emph{elastic part} $\gamma^{\rm e}$ and a \emph{plastic part} $\gamma^{\rm p}$. This decomposition is accompanied by the prescription that, given any part $P=(a,b)\subset (-h,+h)$, the internal power expended within $P$ has the form:
\begin{equation}\label{eq:38}
  \mathscr W_{\rm int}(P)=\int_P \tau\dot \gamma^{\rm e}+\tau^{\rm p}\dot\gamma^{\rm p}+k^{\rm p}\dot\gamma^{\rm p}_y{\rm d}y.
\end{equation}
Thus, power expenditure by the \emph{macroscopic shear stress} $\tau$ is accompanied by working of the \emph{plastic microstress} $\tau^{\rm p}$ and \emph{gradient microstress} $k^{\rm p}$. If body forces are left out of the picture, the external power expended on $P=(a,b)$ is localized on the boundary $\partial P=\{a,b\}$ and has the form:
\begin{equation*}%\label{eq:41}
  \mathscr W_{\rm ext}(P)=\widehat\tau(b)\dot u(b)+\widehat k^{\rm p}(b)\dot\gamma^{\rm p}(b)-\widehat\tau(a)\dot u(a)-\widehat k^{\rm p}(a)\dot\gamma^{\rm p}(a),
\end{equation*}
where $\widehat\tau$ and $k^{\rm p}$ are, respectively, the \emph{macroscopic} and the \emph{microscopic} shear tractions. The application of the principle of virtual powers yields:

\smallskip

\begin{itemize}

\item[1)] the identification between stress and  traction
\begin{equation*}%\label{eq:3}
\tau=\widehat\tau,
\end{equation*}
along with the \emph{macroscopic-force balance}:
\begin{equation}\label{eq:68}
  \tau_y=0;
\end{equation}
\item[2)] the identification of $\widehat k^{\rm p}$ with $k^{\rm p}$, along with the \emph{microscopic force-balance}:
\begin{equation*}%\label{eq:67}
  \tau=\tau^{\rm p}- k^{\rm p}_{y}.
\end{equation*}
\end{itemize}

\subsection{Constitutive prescriptions}

Consistent with the choice \eqref{eq:38} for the internal power expenditure, it is assumed in Ref.~\cite{AnandGLG2005JMPS} that the \emph{free-energy density} $\varphi$ depends on the triplet $(\gamma^{\rm e},\gamma^{\rm p},\gamma^{\rm p}_y)$ through a  constitutive equation of the form:
  \begin{equation*}
    \varphi=\widehat\varphi(\gamma^{\rm e},\gamma^{\rm p},\gamma^{\rm p}_y).
  \end{equation*}
It is also assumed that the constitutive mapping delivering the free-energy density is the sum:
\begin{equation*}
\widehat\varphi(\gamma^{\rm e},\gamma^{\rm p},\gamma^{\rm p}_y)=\widehat\varphi^{\rm e}(\gamma^{\rm e})+\widehat\varphi^{\rm p}(\gamma^{\rm p},\gamma^{\rm p}_{y})
\end{equation*}
of an \emph{elastic-energy mapping} $\widehat\varphi^{\rm e}$, which takes into account the elastic shear, and a \emph{defect-energy mapping} $\hat\varphi^{\rm p}$, which depends on the plastic shear and on its gradient. In particular, the elastic-energy mapping is given the form $\widehat\varphi^{\rm e}(\gamma^{\rm e})=\frac 12 G(\gamma^{\rm e})^2$, with $G>0$ the \emph{shear modulus}. This assumption is accompanied by the standard constitutive prescription $\tau=\frac{\partial \widehat \varphi^{\rm e}}{\partial{\gamma^{\rm e}}}$, whence:
\begin{equation}\label{eq:6}
\tau=G\gamma^{\rm e}.
\end{equation}
The microstresses are then split into an energetic part and a dissipative part by setting
\begin{equation*}
  \tau^{\rm p}=\tau^{\rm dis}+\tau^{\rm en},\qquad k^{\rm p}=k^{\rm dis}+k^{\rm en},
\end{equation*}
where
\begin{equation*}
  \tau^{\rm en}=\frac{\partial\widehat\varphi^{\rm p}}{\partial\gamma^{\rm p}},\qquad k^{\rm en}=\frac{\partial\widehat\varphi^{\rm p}}{\partial\gamma^{\rm p}_y},
\end{equation*}
so that the following reduced form of the dissipation inequality is arrived at:
\begin{equation*}
  0\le \tau^{\rm dis}\dot\gamma^{\rm p}+k^{\rm dis} \dot\gamma^{\rm p}_{y}.
\end{equation*}
By analogy with the constitutive equations describing viscoplastic behavior in metals, the following constitutive equations have been considered in Ref.~\cite{AnandGLG2005JMPS}:
\begin{equation}\label{eq:49}
\begin{aligned}
&\tau^{\rm p}=S\left(\frac{d^{\rm p}}{d_0}\right)^m
\frac{\dot\gamma^{\rm p}}{d^{\rm p}},\qquad   d^{\rm p}=\sqrt{(\dot\gamma^{\rm p})^2+\ell^2\left(\dot\gamma^{\rm p}_{y}\right)^2},\\
&\tau^{\rm dis}=S L^2 \gamma^{\rm p}_{y},\qquad k^{\rm dis}=S_0\ell^2\left(\frac{d^{\rm p}}{d_0}\right)^m
\frac{\dot\gamma^{\rm p}_y}{d^{\rm p}},\\
& \dot S=H(S)d^{\rm p},\qquad S(0)=S_0.\\
\end{aligned}
\end{equation}
Here: $S$ is the \emph{current yield strength}, an internal variable whose value at time $t=0$ is equal to the \emph{initial yield strength} $S_0$ and whose time derivative is proportional to the \emph{effective flow rate} $d^{\rm p}$ through a (isotropic) \emph{hardening/softening function} $H(S)$; $d_0$ is the \emph{reference flow rate}; $m>0$ is the \emph{rate-sensitivity parameter}.

\smallskip

The constitutive prescription (\ref{eq:21b}) follows by setting $H(S)=0$ (no isotropic hardening) and by formally letting $m\to 0$ in \eqref{eq:49} (rate-independent limit). The partial differential equation (\ref{eq:21}a) is recovered by choosing
\begin{equation*}%\label{eq:52}
\widehat\varphi^{\rm p}
(\gamma^{\rm p},\gamma^{\rm p}_y)=\frac 12S_0\big(\kappa(\gamma^{\rm p})^2+L^2(\gamma^{\rm p}_y)^2\big).
\end{equation*}

\subsection{The traction problem}

In the \emph{traction problem}, the bottom side of the strip is clamped, that is,
\begin{equation}\label{eq:7}
  u(-h,t)=0,
\end{equation}
and a time-dependent shear traction $\widehat\tau_h(t)$ is prescribed on the upper side, that is,
\begin{equation*}
  \tau(h,t)=\widehat\tau_h(t).
\end{equation*}
On recalling that the shear stress is spatially constant by \eqref{eq:68}, we see that the shear stress $\tau(t)$ appearing in the flow rule \eqref{eq:21} is a prescribed, spatially--constant field. Thus, the flow rule \eqref{eq:49} can be solved for the plastic shear $\gamma^{\rm p}$ without knowing the displacement field. The latter is recovered by integrating \eqref{eq:5} and \eqref{eq:6}, and by taking \eqref{eq:7} into account, that is to say,
\begin{equation*}
%  \label{eq:55}
  u(y,t)= \int_{-h}^y \left(\frac {\tau(t)} G+ \gamma^{\rm p}(s,t)\right){\rm d}s.
\end{equation*}

{\bf{Acknowledgements.}} G.T. thanks Lallit Anand, Lorenzo Bardella, and Patrizio Neff for stimulating discussions on strain-gradient plasticity.

\bibliographystyle{abbrv}
\bibliography{bibliography}

\end{document}